\documentclass[12pt,3p,sort&compress]{elsarticle}

\usepackage{amsmath,amsthm,amssymb}
\usepackage[mathlines,pagewise]{lineno}
\usepackage{listings}
\usepackage[draft]{varioref}
\usepackage{caption}

\DeclareMathOperator{\Suff}{Suff}
\DeclareMathOperator{\Pref}{Pref}
\DeclareMathOperator{\Fact}{Fact}
\renewcommand{\epsilon}{\varepsilon}
\newcommand{\St}{\mathit{St}}

\newcommand{\oc}{\mathrm{oc}}
\newcommand{\wt}{\widetilde}

\renewcommand{\epsilon}{\varepsilon}

\newtheorem{theorem}{Theorem}%[section]
\newtheorem{proposition}[theorem]{Proposition}
\newtheorem{lemma}[theorem]{Lemma}
\newtheorem{corollary}[theorem]{Corollary}

\theoremstyle{definition}
\newtheorem{definition}[theorem]{Definition}
\newtheorem{example}[theorem]{Example}

\newtheorem{remark}[theorem]{Remark}

\begin{document}

%\linenumbers

\sloppy

\begin{frontmatter}

\title{The sequence of open and closed prefixes of a Sturmian word \tnoteref{note1}}
\tnotetext[note1]{Some of the results contained in this paper were presented at the 9th International Conference on Words, WORDS~2013 \cite{DelFi13}.}
 
\author{Alessandro De Luca}
 \ead{alessandro.deluca@unina.it}
% \ead[url]{home page}
 \address{DIETI, Universit\`a degli Studi di Napoli Federico II, Italy}
 
\author{Gabriele Fici\corref{cor1}}
\ead{gabriele.fici@unipa.it}
% \ead[url]{home page}
\address{Dipartimento di Matematica e Informatica, Universit\`a di Palermo, Italy}

\author{Luca Q. Zamboni}
\ead{lupastis@gmail.com}
% \ead[url]{home page}
\address{Universit\'e Claude Bernard Lyon 1, France}

\cortext[cor1]{Corresponding author.}

\journal{Advances in Applied Mathematics}

\begin{abstract}
  A finite word is closed if it contains a factor that occurs both as a prefix and as a suffix but does not have internal occurrences, otherwise it is open. 
  We are interested in the {\it oc-sequence} of a word, which is the binary sequence whose $n$-th element is $0$ if the prefix of length $n$ of the word is open, or $1$ if it is closed. We exhibit results showing that this sequence is deeply related to the combinatorial and periodic structure of a word. In the case of Sturmian words, we show that these are uniquely determined (up to renaming letters) by their oc-sequence. Moreover, we prove that the class of finite Sturmian words is a maximal element with this property in the class of binary factorial languages. We then discuss several aspects of Sturmian words that can be expressed through this sequence. Finally, we provide a linear-time algorithm that computes the oc-sequence of a finite word, and a linear-time algorithm that reconstructs a finite Sturmian word from its oc-sequence.
\end{abstract}


\end{frontmatter}

%%%%%%%%%%%%%%%%%%%%%%%%%%%%%%%%%%%%%%%%%%%%%%%%%%%%%%%%%%%%%%%%%%%%%%%%%%%%%%%%%%%%%%%%%%%
\section{Introduction}
%%%%%%%%%%%%%%%%%%%%%%%%%%%%%%%%%%%%%%%%%%%%%%%%%%%%%%%%%%%%%%%%%%%%%%%%%%%%%%%%%%%%%%%%%%%

In a recent paper with M.~Bucci~\cite{BuDelFi13}, the first two authors dealt with trapezoidal words (a generalization of finite Sturmian words), also with respect to the property of being closed or open.  Let $\Sigma$ be a finite nonempty set (the alphabet). A (finite)  word $w=w[1]w[2]\cdots w[n]$ with $w[i] \in \Sigma$ is  \emph{closed} (also known as \emph{periodic-like}~\cite{CaDel01a}) if it contains a factor that occurs both as a prefix and as a suffix but does not have internal occurrences, otherwise it is \emph{open}. For example, the words $abca$, $ababa$ and $aabaab$ are closed --- any word of length $1$ is closed, the empty word being a factor that occurs both as a prefix and as a suffix but does not have internal occurrences; the words $ab$, $aab$ and $aaba$, instead, are open.

Given a finite or infinite word $w=w[1]w[2]\cdots$, the sequence $\oc(w)$ of open/closed prefixes of $w$, that we refer to as  the \emph{oc-sequence} of $w$, is the binary sequence $c(1)c(2)\cdots$ whose $n$-th element is $1$ if the prefix of $w$ of length $n$ is closed, $0$ if it is open. For example, if $w=abcab$, then $\oc(w)=10011$.

A question that arises naturally  is whether it is possible to reconstruct a word (up to renaming letters) from its oc-sequence. This is not true in general, even when the alphabet is binary. For example, the words $aaba$ and $aabb$ are not isomorphic (i.e., one cannot be obtained from the other by renaming letters), yet they have the same oc-sequence $1100$. As a first result of this paper, we show that if a word is known to be Sturmian, then it can be reconstructed (up to renaming letters) from its oc-sequence. That is, Sturmian words are characterized by their oc-sequences. Moreover, we prove that the class of finite Sturmian words is a maximal element with this property in the class of binary factorial languages. %That is, if two non-isomorphic words have the same sequence of open/closed prefixes, they cannot be both Sturmian.

In~\cite{BuDelFi13}, the authors investigated the structure of the sequence $\oc(F)$ of the Fibonacci word $F$. They proved that the lengths of the runs (maximal subsequences of consecutive equal elements) in $\oc(F)$ form the doubled Fibonacci sequence. We prove in this paper that this doubling property holds for every standard Sturmian word, and describe the sequence $\oc(w)$ of a standard Sturmian word $w$ in terms of the \emph{semicentral} prefixes of $w$, which are the prefixes of the form $u_{n}xyu_{n}$, where $x,y$ are letters and $u_{n}xy$ is an element of the standard sequence of $w$. 
As a consequence, we show that the word $ba^{-1}w$, obtained from a standard Sturmian word $w$ starting with letter $a$ by replacing the first letter with a $b$, can be written as the infinite product of the words $(u_{n}^{-1}u_{n+1})^{2}$, $n\ge 0$. Since the words $u_{n}^{-1}u_{n+1}$ are reversals of standard words, this induces an infinite factorization of $ba^{-1}w$ in squares of reversed standard words. 

We then show how the oc-sequence of a standard Sturmian word of slope $\alpha$ is related to the continued fraction expansion of $\alpha$, both in terms of the convergents and of the continuants of $\alpha$. 

Finally, we provide a linear-time algorithm that computes the oc-sequence of a finite word, and a linear-time algorithm that reconstructs a finite Sturmian word from its oc-sequence.

%%%%%%%%%%%%%%%%%%%%%%%%%%%%%%%%%%%%%%%%%%%%%%%%%%%%%%%%%%%%%%%%%%%%%%%%%%%%%%%%%%%%%%%%%%%
\section{Open and closed words}
%%%%%%%%%%%%%%%%%%%%%%%%%%%%%%%%%%%%%%%%%%%%%%%%%%%%%%%%%%%%%%%%%%%%%%%%%%%%%%%%%%%%%%%%%%%

Let us begin with some notation and basic definitions; for those not included below, we refer the reader to~\cite{BuDelFi13} and~\cite{LothaireAlg}.

Let $\Sigma$ be a finite alphabet. Let $\Sigma^{*}$ and $\widehat{\Sigma}^{*}$ stand respectively for the free monoid and the free group generated by $\Sigma$. Their elements are called \emph{words} over $\Sigma$. The \emph{length} of a word $w$ is denoted by $|w|$. The \emph{empty word}, denoted by $\epsilon$, is the unique word of length zero and is the neutral element of $\Sigma^{*}$ and $\widehat{\Sigma}^{*}$.
If $x\in \Sigma$ and $w\in \Sigma^*$, we let $|w|_x$ denote the number of occurrences of $x$ in $w$. 
 
A \emph{prefix} (resp.~a \emph{suffix}) of a word $w$ is any word $u$ such that $w=uz$ (resp.~$w=zu$) for some word $z$. A \emph{factor} of $w$ is a prefix of a suffix (or, equivalently, a suffix of a prefix) of $w$.  A prefix/suffix/factor of a word is \emph{proper} if it is nonempty and does not coincide with the word itself. The set of prefixes, suffixes and factors of the word $w$ are denoted  by $\Pref(w)$, $\Suff(w)$ and $\Fact(w)$, respectively.
From the definitions, we have that $\epsilon$ is a prefix, a suffix and a factor of any word. 
A \emph{border} of a word $w$ is any word in $\Pref(w)\cap \Suff(w)$ different from $w$. 
An \emph{occurrence} of a factor $u$ in $w$ is a factorization $w=vuz$. An occurrence of $u$ is \emph{internal} if both $v$ and $z$ are nonempty. 

A  \emph{period} of a nonempty word $w$ is an integer of the form $|w|-|u|$, where $u$ is a border of $w$. We call \emph{the} period of $w$ the least of its periods, that is the difference between the length of $w$ and the length of its longest border. Conventionally, the period of $\varepsilon$ is 1. The ratio between the length and the period of a word $w$ is called the \emph{exponent} of $w$. %A word is called  a \emph{power} if it has integer exponent greater than $1$.

A factor $v$ of a word $w$ is \emph{left special in $w$} (resp.~\emph{right special in $w$}) if there exist $a,b\in\Sigma$ such that $av$ and $bv$ are factors of $w$ (resp.~$va$ and $vb$ are factors of $w$). A \emph{bispecial factor} of $w$ is a factor that is both left and right special.

The word $\wt{w}$ obtained by reading $w$ from right to left is called the \emph{reversal} (or \emph{mirror image}) of $w$. A \emph{palindrome} is a word $w$ such that $\wt{w}=w$.
In particular, the empty word is  a palindrome. 

An \emph{infinite word} $w$ over $\Sigma$ is a sequence $w:\mathbb N_+\to\Sigma$, written as $w=w[1]w[2]\dotsm w[n]\dotsm$. Prefixes and factors of infinite words are naturally defined, as is the product $uw$ of a finite word $u$ and an infinite word $w$.
Let $\Sigma^{\omega}$ denote the set of infinite words over $\Sigma$. If $u$ is a finite nonempty word, $u^\omega$ denotes the periodic word $uuu\dotsm\in \Sigma^{\omega}$. An infinite word $w$ is said to be \emph{ultimately periodic} if there exist two finite words $v$ and $u$ such that $w=vu^\omega$; an \emph{aperiodic} word is an infinite word that is not ultimately periodic.
An infinite word $w$ is \emph{recurrent} if every factor of $w$
occurs infinitely often; equivalently, $w$ is recurrent if and only if every prefix of $w$ has a second occurrence in $w$.

We recall the definitions of open and closed words given in \cite{Fi11}:

\begin{definition}\label{def:closed}
A finite word $w$ is \emph{closed} if it is empty or has a factor $v\neq w$ occurring exactly twice in $w$, as a prefix and as a suffix of $w$ (with no internal occurrences). A word that is not closed is called \emph{open}.
\end{definition}

%The word $aba$ is closed, since its factor $a$ appears only as a prefix and as a suffix. The word $abaa$, on the contrary, is not closed. 
For any letter $a\in \Sigma$ and for any $n>0$, the word $a^{n}$ is closed, $a^{n-1}$ being a factor occurring only as a prefix and as a suffix in it (this includes the special case of single letters, for which $n=1$ and $a^{n-1}=\epsilon$). 
More generally, every word whose exponent is at least $2$ is closed \cite[Proposition~4]{BaFiLi15}.

\begin{remark}
The notion of closed word is equivalent to that of \emph{periodic-like} word \cite{CaDel01a}. A word $w$ is periodic-like if its longest repeated prefix is not right special.

The notion of closed word is also closely related to the concept of \emph{complete return} to a factor, as considered in \cite{GlJuWiZa09}. A complete return to the factor $u$ in a word $w$ is any factor of $w$ having exactly two occurrences of $u$, one as a prefix and one as a suffix. Hence, $w$ is closed if and only if it is a complete return to one of its factors; such a factor is clearly both the longest repeated prefix and the longest repeated suffix of $w$ (i.e., the longest border of $w$). 
\end{remark}
 
\begin{remark}\label{obs}
Let $w$ be a nonempty word over $\Sigma$. The following characterizations of closed words follow easily from the definition:
 
\begin{enumerate}
 \item the longest repeated prefix (resp.~suffix) of $w$ does not have internal occurrences in $w$, i.e., occurs in $w$ only as a prefix and as a suffix;
  %\item the longest repeated suffix of $w$ does not have internal occurrences in $w$, i.e., occurs in $w$ only as a suffix and as a prefix;
  \item the longest repeated prefix (resp.~suffix) of $w$ is not a right (resp.~left) special factor of $w$;
  %\item the longest repeated suffix of $w$ is not a left special factor of $w$;
 \item $w$ has a border that does not have internal occurrences in $w$;
 \item the longest border of $w$ does not have internal occurrences in $w$.
\end{enumerate}

Obviously, the negations of the previous properties characterize open words. In the rest of the paper we will use these characterizations freely and without explicit mention to this remark.
\end{remark}

We conclude this section with some results on right extensions.

\begin{lemma}\label{lem:nbo}
 Let $w$ be a nonempty word over $\Sigma$, and $x\in\Sigma$ be such that
 $wx$ is closed. Then $wx$ has the same period as $w$.
\end{lemma}

\begin{proof}
Let $vx$ be the longest border of $wx$, and $v'$ be the longest border of $w$. By contradiction, suppose $|v'|>|v|$. Then $vx$ is a prefix of $v'$, and therefore has an internal occurrence in $wx$, contradicting the hypothesis that $wx$ is closed. Hence, $v$ is the longest border of $w$, so that $w$ and $wx$ have the same period $|w|-|v|$.
\end{proof}

\begin{lemma}\label{cor:nbo}
For all nonempty $w\in\Sigma^{*}$, there exists at most one letter $x\in\Sigma$ such that $wx$ is closed.
\end{lemma}
\begin{proof}
Straightforward after Lemma~\ref{lem:nbo}.
\end{proof}

If $w$ is closed, then exactly one such extension is closed.
More precisely, we have the following result (see also~\cite[Prop.~4]{CaDel01a}).

\begin{lemma}\label{lem:ce}
 Let $w$ be a closed word. Then $wx$, $x\in \Sigma$, is closed if and only if $wx$ has the same period as $w$.
\end{lemma}

\begin{proof}
The case $w=\varepsilon$ is trivially verified, so
let $w$ be a nonempty closed word and $v$ be its longest border%; in particular, $v$ is the longest repeated prefix of $w$
. Let $x$ be the letter such that $wx$ is has the same period as $w$, i.e., such that $vx$ is a prefix of $w$. Then $wx$ is closed, as its border $vx$ cannot have internal occurrences.
The converse follows from Lemma~\ref{lem:nbo}.
\end{proof}

For more details on open and closed words and related results the reader can see~\cite{CaDel01a,BuDelDel09,Fi11,BuDelFi13,Ba+16}.

%%%%%%%%%%%%%%%%%%%%%%%%%%%%%%%%%%%%%%%%%%%%%%%%%%%%%%%%%%%%%%%%%%%%%%%%%%%%%%%%%%%%%%%%%%%
\section{The oc-sequence of a word}
%%%%%%%%%%%%%%%%%%%%%%%%%%%%%%%%%%%%%%%%%%%%%%%%%%%%%%%%%%%%%%%%%%%%%%%%%%%%%%%%%%%%%%%%%%%

We now define the oc-sequence of a word.

\begin{definition}
 Let $w=w[1]w[2]\dotsm w[n]\dotsm$ be a finite or infinite word over $\Sigma$. We define $\oc(w)=c(1)c(2)\dotsm c(n)\dotsm$, called the \emph{oc-sequence} of $w$, as the binary sequence whose $n$-th element is $0$ if the prefix of length $n$ of $w$ is open, or $1$ if it is closed.
\end{definition}
For example, if $w=abaaab$, then $\oc(w)=101001$.

\begin{remark}
\label{rem:nth1}
By definition of closed word, for each integer $n\geq 1$, the $(n+1)$-st occurrence of 1 in $\oc(w)$ is at the position corresponding to the end of the second occurrence of the prefix of length $n$ in $w$. Hence, if a finite word $w$ admits a border of length $\ell$, then
$|\oc(w)|_1\geq \ell+1.$ 

In particular, a closed word $w$ is a complete return to its prefix of length $|\oc(w)|_1-1$; equivalently, the period of a closed word $w$ is equal to $1+|\oc(w)|_0$.
\end{remark}

In the following two propositions we relate recurrence and periodicity of an infinite word with analogous properties of its oc-sequence.

\begin{proposition}
\label{prop:arrayric}
Let $w\in\Sigma^\omega$. The following are equivalent:
\begin{enumerate}
    \item\label{arrayric} $\oc(w)$ is recurrent;
    \item\label{costante} $w=x^\omega$ for a letter $x\in \Sigma$;
    \item\label{arraycost} $\oc(w)=1^\omega$.
\end{enumerate}
\end{proposition}
\begin{proof}
Clearly, $\ref{costante}\Leftrightarrow\ref{arraycost}\Rightarrow\ref{arrayric}$. To complete the proof, we show that $\ref{arrayric}\Rightarrow\ref{arraycost}$. Let then $\oc(w)$ be recurrent, and suppose by contradiction that $0$ occurs in it. Thus, there exists a positive integer $t$ such that $10^t1$ occurs infinitely often in $\oc(w)$. Hence, for every $n\geq t$, there exists $P$ such that $P10^t1$ is a prefix of $\oc(w)$ and $|P|_1=:m\geq n$. Let $u$ be the prefix of length $m$ of $w$; by Remark~\ref{rem:nth1}, we obtain that the prefixes of $w$ of length $|P1|$ and $|P10^t|$ both have $u$ as a suffix. We have found two occurrences of $u$ at distance $t$ from each other, so that $u$ must have $t$ as a period. Since $n$ is arbitrary and $|u|=m\geq n$, it follows that $w$ has period $t$, so that $\oc(w)$ ends in $1^\omega$ as a consequence of Lemma~\ref{lem:ce}. This contradicts the hypothesis that $\oc(w)$ is recurrent and contains $0$.
\end{proof}

\begin{proposition}
Let $w\in\Sigma^\omega$. The sequence $\oc(w)$ is ultimately periodic if and only if $w$ is either periodic or not recurrent. In the first case, $\oc(w)$ ends
in $1^\omega$, while in the latter case it ends  in $0^\omega$.
\end{proposition}

\begin{proof}
The ``if'' part is immediate. Let us then prove the ``only if'' part; let $\oc(w)=UV^\omega$. Suppose first that $1$ does not occur in $V$. Then $\oc(w)$ ends in $0^\omega$, so that $w$ has prefixes that have no other occurrences in $w$; hence, $w$ is not recurrent. If $0$ does not occur in $V$, then $\oc(w)$ ends in $1^\omega$ so that $w$ is periodic as a consequence of Lemma~\ref{lem:ce}. Finally, suppose that both $1$ and $0$ occur in $V$. Then there exists a positive integer $t$ such that $10^t1$ occurs infinitely often in $\oc(w)$; as we have seen in the proof of Proposition~\ref{prop:arrayric}, this leads to a contradiction.
\end{proof}

The following lemma shows that in the sequence $\oc(w)$ any run of $0$s is at least as long as the previous run of $1$s. It will be useful in what follows.

\begin{lemma}
\label{lem:ts}
Given positive integers $s$ and $t$, if $1^t0^s1$ is a factor of $\oc(w)$ then $t\leq s$.
\end{lemma}

\begin{proof}
Let $w=w[1]w[2]w[3]\cdots$ with $w[i]\in\Sigma$, and let
$c=\oc(w)=c(1)c(2)\cdots$ with $c(i)\in\{0,1\}$ for all integers $i\geq 1$.
Let $a\in\Sigma$ be the letter such that $w[1]=a$.
The result is clear in the case when $1^t0^s1$ is a prefix of $\oc(w)$, for this implies that $w$ begins in $a^tb$, where $b$ is a letter in $\Sigma$ different from $a$. Since the longest border of $a^tb$ is the empty word, it follows that the next occurrence of $a^t$ must occur within the suffix $w[t+2]w[t+3]\cdots  $ of $w$, so that $c(t+1)\cdots c(2t)=0^t$ whence $t\leq s$.

We may now assume that $1^t0^s1$ occurs in $c$ at some later position. Fix a positive integer $r$ such that $1^t0^s1$ is a suffix of  $c(1)\cdots c(r+s+1)$. Let $n=|c(1)\cdots c(r)|_1$ and $u$ be the prefix of $w$ of length $n-1$. 
We note that since $1^t0^s1$ occurs in $c$ and not just as a prefix, we have $t<n$ and $n\geq 2$ (hence $u$ is nonempty). It follows that there exist distinct letters $x,y\in \Sigma$ such that $w$ begins in $ux$ and $w[1]\cdots w[r+1]$ terminates in $uy$. Hence, the second occurrence of $u$ in $w$ terminates in position $r$, while the second occurrence of $ux$ in $w$ terminates in position $r+s+1$. If the second occurrence of $ux$ in $w$ does not overlap the second occurrence of $u$ in $w$, then $s\geq |u|=n-1 \geq t$. If the second  occurrence of $ux$ in $w$ overlaps the second occurrence of $u$ in $w$ by an amount $s'\geq 1$, then we have that $s+s'=|u|=n-1$ and $u$ has a border of length $s'$.  
Let $v$ denote the longest border of $u$. Thus $|v|\geq s'$.
First suppose that either $c(|u|)=0$ or $c(|u|)=1$ but $c(|u|)$ and $c(r)$ do not belong to the same run.  Then, 
since $|c(1)\cdots c(|u|)|_1=|(u)|_{1}\geq s'+1$ by Remark~\ref{rem:nth1}, we deduce that
\[t\leq |c(|u|+1)\cdots c(r)|_{1}\leq n-(s'+1)=s,\] as required.

Finally, suppose $c(|u|)=1$ with $c(|u|)$ and $c(r)$ belonging to the same run.
In this case, $u$ and $ux$ are both closed, so that $vx$ is a prefix of $u$.
Therefore $|v|>s'$, since $w[1]\cdots w[s']y$ is a prefix of $u$ as well, and hence $v$ has a border of length $s'$. Now, let $px$ be the prefix of $w$ (and of $vx$) that terminates with the first occurrence of $w[1]\cdots w[s']x$; then $px$ is necessarily open, and $|\oc(p)|_{1}\geq s'+1$ by Remark~\ref{rem:nth1}.
It follows that if
$1^i$ is a suffix of $c(1)\cdots c(|u|)$, hence $i\leq |v|-s'$. Thus, $t\leq |v|-s' + |u|-|v|=|u|-s'=s$. 
\end{proof}

%%%%%%%%%%%%%%%%%%%%%%%%%%%%%%%%%%%%%%%%%%%%%%%%%%%%%%%%%%%%%%%%%%%%%%%%%%%%%%%%%%%%%%%%%%%
\subsection{Sturmian words}
%%%%%%%%%%%%%%%%%%%%%%%%%%%%%%%%%%%%%%%%%%%%%%%%%%%%%%%%%%%%%%%%%%%%%%%%%%%%%%%%%%%%%%%%%%%

  We let $\Sigma=\{a,b\}$ be a fixed binary alphabet from now on, unless otherwise specified.
  An element of $\Sigma^{\omega}$ is a \emph{Sturmian word} if it contains exactly $n+1$ distinct factors of length $n$, for every $n\ge 0$. 
 A famous example of Sturmian word is the Fibonacci word \[F=abaababaabaababaababa\cdots\] that is
 the limit, as $n\to \infty$, of the sequence of words $(f_n)$, called the sequence of \emph{finite Fibonacci words}, defined by $f_{-1}=b$, $f_0=a$ and, for every $n\geq 1$, $f_n=f_{n-1}f_{n-2}$.  

It is well known that if $w$ is a Sturmian word then at least one of $aw$ and $bw$ is also a Sturmian word. A Sturmian word $w$ is called  \emph{standard} (or \emph{characteristic}) if $aw$ and $bw$ are both Sturmian words. The Fibonacci word is an example of standard Sturmian word. In the next section, we will deal specifically with standard Sturmian words. Here, we focus on finite factors of Sturmian words, called \emph{finite Sturmian words}. Actually, finite Sturmian words are precisely the elements of $\Sigma^{*}$ verifying the following balance property:  for any $u,v\in \Fact(w)$ such that $|u|=|v|$ one has $||u|_{a}-|v|_{a}|\le 1$ (or, equivalently, $||u|_{b}-|v|_{b}|\le 1$). 

We let $\St$ denote the set of finite Sturmian words. The language $\St$ is factorial (i.e., if $w=uv\in \St$, then $u,v\in \St$) and extendible (i.e., for every $w\in \St$ there exist letters $x,y\in \Sigma$ such that $xwy\in \St$).

We recall the following definitions given in \cite{DelMi94}. 

\begin{definition}
A word  $w\in \Sigma^{*}$ is a  left special (resp.~right special) Sturmian word if $aw,bw\in \St$ (resp.~if $wa,wb\in \St$). A bispecial Sturmian word is a Sturmian word that is both left special and right special. Moreover, a bispecial Sturmian word is strictly bispecial if $awa,awb,bwa,$ and $bwb$ are all Sturmian words; otherwise it is non-strictly bispecial. 
\end{definition}

For example, the word $w=ab$ is a bispecial Sturmian word, since $aw$, $bw$, $wa$ and $wb$ are all Sturmian. This example also shows that a bispecial Sturmian word is not necessarily a bispecial factor of some Sturmian word (which must be a palindrome); in fact, bispecial factors of Sturmian words coincide with \emph{strictly} bispecial Sturmian words (see \cite{Fi14} for more details on bispecial Sturmian words).

\begin{remark}\label{rem:rsp}
It is known that if $w$ is a left special Sturmian word, then $w$ is a prefix of some standard Sturmian word, and the left special factors of $w$ are prefixes of $w$. Symmetrically, if $w$ is a right special Sturmian word, then the right special factors of $w$ are suffixes of $w$. 
\end{remark}

Regarding open and closed prefixes of Sturmian words, we prove the following result.

\begin{theorem}\label{theor:main}
 Every (finite or infinite) Sturmian word $w$ is uniquely determined, up to isomorphisms of the alphabet $\Sigma$, by its oc-sequence $\oc(w)$.
\end{theorem}

We need some intermediate lemmas.

\begin{lemma}\label{lem:lsp}
Let $w$ be a right special Sturmian word and let $u$ be its longest repeated prefix. Then $u$ is a suffix of $w$.
\end{lemma}

\begin{proof}
If $w$ is closed, the claim follows from the definition of closed word. If $w$ is open, then $u$ is right special in $w$, and by
Remark~\ref{rem:rsp} $u$ is a suffix of $w$. 
\end{proof}

\begin{lemma}\label{lem:speclo}
 Let $w$ be a right special Sturmian word.  Then $wa$ or $wb$ is closed.
\end{lemma}

\begin{proof}
Let $u$ be the longest repeated prefix of $w$ and $x$ be the letter following the occurrence of $u$ as a prefix of $w$. By Lemma \ref{lem:lsp}, $u$ is a suffix of $w$. Clearly, the longest repeated prefix of $wx$ is $ux$, which is also a suffix of $wx$ and cannot have internal occurrences in $wx$, otherwise the longest repeated prefix of $w$ would not be $u$. Therefore, $wx$ is closed.
\end{proof}

So, by Lemmas~\ref{cor:nbo} and~\ref{lem:speclo}, if $w$ is a right special Sturmian word, then one of $wa$ and $wb$ is closed and the other is open. This implies that the oc-sequence of a (finite or infinite) Sturmian word characterizes it up to exchange of letters. The proof of Theorem \ref{theor:main} is therefore complete.

%%BEGIN C* STUFF
We now prove that $St$ is maximal in the class of factorial languages over $\Sigma$ verifying the
condition of Theorem~\ref{theor:main}, i.e., such that their members are determined by their
$\oc$ sequences. Let us write $u\sim v$ when two words $u,v\in\Sigma^{*}$ are isomorphic, and let
\[\mathcal{C}=\{A\subseteq \Sigma^*\mid \forall u\in A, \Fact(u)\subseteq A
\,\wedge\, \forall u,v \in A: \oc(u)=\oc(v)\Rightarrow u\sim v\}.\]
We note that $\mathcal{C}$ is nonempty (e.g., $A=\{\varepsilon,0\}\in\mathcal{C}$), partially ordered with respect to inclusion, and such that every increasing chain 
\[A_1\subseteq A_2 \subseteq A_3\subseteq \cdots\]
with all $A_i\in \mathcal{C}$
has an upper bound in $\mathcal{C}$ given by  $\bigcup_{i\geq 1}A_i.$ Thus, by Zorn's lemma, $\mathcal{C}$ admits at least one maximal element. 

\begin{theorem}
\label{thm:C}
$St$ is a maximal element of $\mathcal{C}.$
\end{theorem}

Again we need to recall two lemmas. The first is a well-known result about balanced words 
(cf.~\cite[Proposition~2.1.3]{LothaireAlg}):
\begin{lemma}
\label{thm:unbal}
A word $s\in\Sigma^{*}$ is not balanced if and only if there exists a palindrome $v$ such that
$ava,bvb\in\Fact(s)$.
\end{lemma}

Next is an immediate consequence of 
known properties of Christoffel words (cf.~\cite{Fi14}).
\begin{lemma}
\label{thm:bispCF}
A word $u\in\Sigma^{*}$ is a non-strictly bispecial Sturmian word if and only if there exists
a strictly bispecial Sturmian word $w$ and an integer $n>1$ such that
\[\text{\emph{either} }\;aub=(awb)^{n}\in St\;
\text{ \emph{or} }\; bua=(bwa)^{n}\in St\,.\]
\end{lemma}

\begin{proof}[Proof of Theorem~\ref{thm:C}]
It follows from Theorem~\ref{theor:main} that  $St \in \mathcal{C}.$ To see that $St$ is a maximal element of $\mathcal{C}$  we show that no element of $\mathcal{C}$ properly contains $St$. Suppose to the contrary that there exists an element $A\in \mathcal{C}$  such that $St\subsetneq  A$.
Let $s$ be an element of minimal length of $A$ not belonging to $St$.
By Lemma~\ref{thm:unbal}, there exists a word $v$ such that $ava,bvb\in\Fact(s)$. Since
all proper factors of $s$ are balanced, without loss of generality we can assume that $ava$ is
a prefix of $s$ and $bvb$ is a suffix.
Hence we can write $s=aub$ for some $u\in \Sigma^+$.

Let $r$ be a border of $s$. Since $r$ is balanced, we have $|r|<|ava|=|bvb|$. Writing $ava=r\alpha$ and $bvb=\beta r$, it follows that $|\alpha|=|\beta|$ and
$|\alpha|_{a}-|\beta|_{a}=2$, whence $r=\varepsilon$ by our minimality assumption on $s$.
Therefore $s$ is open, so that $\oc(s)$ terminates in $0$. We will show that $aua\in St$ and $\oc(aua)$ terminates in $0$. It follows then that $aua, s\in A$ and that $\oc(aua)=\oc(s)$, a contradiction since $aua\not\sim s$.

By definition of $\mathcal{C}$ it follows that $au, ub\in A$.
By minimality of the length of $s$ we have $au, ub \in St$. Thus $aua$ and $bub\in St$, so that $ua, ub, au, bu \in St$; in other words, $u$ is a bispecial Sturmian word. On the other hand, as $s=aub\notin St$, we have that $u$ is non-strictly bispecial. Thus, by
Lemma~\ref{thm:bispCF}, there exists a word $w$ such that $bua=(bwa)^n$ for some $n>1$. 
Hence $aua=awa(bwa)^{n-1}$. 
Clearly, $awa$ occurs only once in $aua$, as all other factors of the same length have one less occurrence of the letter $a$. Thus, if $z$ is a border of $aua$, then $|z|<|awa|$.  It follows that $z$ is a proper suffix of $bwa$ and so it has an internal occurrence in $aua$ (as a proper suffix of  $awa$). Therefore $aua$ is open, so that $\oc(aua)$ terminates in $0$, as required.
\end{proof}

%%%%%%%%%%%%%%%%%%%%%%%%%%%%%%%%%%%%%%%%%%%%%%%%%%%%%%%%%%%%%%%%%%%%%%%%%%%%%%%%%%%%%%%%%%%
\subsection{Standard Sturmian words}
%%%%%%%%%%%%%%%%%%%%%%%%%%%%%%%%%%%%%%%%%%%%%%%%%%%%%%%%%%%%%%%%%%%%%%%%%%%%%%%%%%%%%%%%%%%

%FORMULA
%Let $r_n=\wt{s_n}$ for all $n\geq -1$, so that $r_{-1}=b$, $r_0=a$, and
%$r_{n+1}=r_{n-1}r_n^{d_n}$ for $n\geq 0$.
%\[\begin{split}
% w&=\prod_{n=0}^\infty r_n^{d_n}=
%    \prod_{n=0}^\infty r_n r_{n-1}^{-1}r_{n-1}r_n^{d_n}r_n^{-1}=
%    \prod_{n=0}^\infty r_n r_{n-1}^{-1}r_{n+1}r_n^{-1}\\
%  &=r_0 r_{-1}^{-1}\prod_{n=0}^\infty \left(r_{n+1}r_n^{-1}\right)^2
%   =ab^{-1}\prod_{n=0}^\infty \left(r_{n+1}r_n^{-1}\right)^2\;.
%  \end{split}\]

In this section, we deal with the oc-sequence of  standard Sturmian words. In \cite{BuDelFi13} a characterization of the oc-sequence of the Fibonacci word $F$ was given.

Let us begin by recalling some definitions and basic results about standard Sturmian words. For more details, the reader can see \cite{Be07} or \cite{LothaireAlg}.

Let $\alpha$ be an irrational number such that $0<\alpha<1$, and let $\left[0;d_{0}+1,d_{1},\ldots\right]$ be the continued fraction expansion of $\alpha$.
The sequence of words defined by $s_{-1}=b$, $s_{0}=a$ and $s_{n+1}=s_{n}^{d_{n}}s_{n-1}$ for $n\ge 0$, converges to the infinite word $w_{\alpha}$, called the \emph{standard Sturmian word of slope $\alpha$}. % this is because $\alpha=\lim_{n\to \infty}h(n)/n$, where $h(n)$ is the number of $b$'s in the prefix of length $n$ of $w$. 
The sequence of words $s_{n}$ is called the \emph{standard sequence} of $w_{\alpha}$. 

Note that $w_{\alpha}$ starts with letter $b$ if and only if $\alpha>1/2$, i.e., if and only if $d_{0}=0$. In this case, $\left[0;d_{1}+1,d_{2},\ldots\right]$ is the continued fraction expansion of $1-\alpha$, and $w_{1-\alpha}$ is the word obtained from $w_{\alpha}$ by exchanging $a$'s and $b$'s. Hence, without loss of generality, we will suppose in the rest of the paper that $w$ starts with letter $a$, i.e., that $d_{0}>0$. 

For every $n\ge -1$, one has 
\begin{equation}\label{eq:su}
s_{n}=u_{n}xy, 
\end{equation}
for $x,y$ letters such that $xy=ab$ if $n$ is odd or $ba$ if $n$ is even. Indeed, the sequence $(u_{n})_{n\ge -1}$ can be defined by: $u_{-1}=a^{-1}$, $u_0=b^{-1}$, and, for every $n\ge 1$, 
\begin{equation}\label{eq:un+1}
u_{n+1}=(u_{n}xy)^{d_{n}}u_{n-1}\,,
\end{equation}
where $x,y$ are as in \eqref{eq:su}.

\begin{example}
 The Fibonacci word $F$ is the standard Sturmian word of slope $1/\varphi^2=(3-\sqrt{5})/2$, whose continued fraction expansion is $[0;2,1,1,1,\ldots]$, so that $d_{n}=1$ for every $n\ge 0$. Therefore, the standard sequence of the Fibonacci word $F$ is the sequence $(f_n)$ defined by: $f_{-1}=b$, $f_{0}=a$, $f_{n+1}=f_{n}f_{n-1}$ for $n\ge 0$. This sequence is the sequence of finite Fibonacci words.
\end{example}

\begin{definition}
 A \emph{standard word} is a finite word belonging to some standard sequence. A \emph{central word} is a word $u\in\Sigma^*$ such that $uxy$ is a standard word, for letters $x,y\in \Sigma$. 
\end{definition}

It is known that every central word is a palindrome%
%(we assume here that $w^{-1}$ is a palindrome if $w$ is a palindrome)
. Actually, central words play a central role in the combinatorics of Sturmian words and have several combinatorial characterizations (see \cite{Be07} for a survey). We summarize some of these properties in the following proposition.
\begin{proposition}
\label{prop:PER}
Let $v$ be a word over $\Sigma$. The following are equivalent:
\begin{enumerate}
    \item $v$ is a central word;
    \item $v$ is a palindromic bispecial Sturmian word;
    \item $v$ is a power of a single letter or it can be written as
    \[v=pxyq=qyxp\] for some words $p$ and $q$ and distinct letters $x,y$.
\end{enumerate}
Moreover, in this latter case, $p$ and $q$ are central words themselves, and $v$ is a complete return to the longest between $p$ and $q$. In particular, central words are closed.
\end{proposition}

In fact, all Sturmian palindromes (and more generally, all \emph{rich} palindromes~\cite{GlJuWiZa09}) are closed; however, in general there do exist open palindromes, such as $aabbabaaababbaa$ (cf.~\cite[Remark 4.13]{BuDelFi13}).
\begin{remark}
Let $(s_{n})_{n\ge -1}$ be a standard sequence. It follows by the definition that for every $k\ge 0$ and $n\ge -1$, the word $s_{n+1}^{k}s_{n}$ is a standard word. In particular, for every $n\ge -1$, the word $s_{n+1}s_n=u_{n+1}yxu_nxy$ is a standard word. Therefore, for every $n\ge -1$, we have that 
\begin{equation}\label{eq:u}
u_nxyu_{n+1}=u_{n+1}yxu_n
\end{equation} is a central word.
\end{remark}

The following lemma is a well-known result (cf.~\cite{fisch}).

\begin{lemma}\label{lem:pref}
Let $w$ be a standard Sturmian word and let $(s_{n})_{n\ge -1}$ be its standard sequence. Then:
\begin{enumerate}
 \item
A standard word $v$ is a prefix of $w$ if and only if $v=s_{n}^{k}s_{n-1}$, for some $n\ge 0$ and $k\le d_{n}$.
 \item  A central word $u$ is a prefix of $w$ if and only if $u=(u_{n}xy)^{k}u_{n-1}$, for some $n\ge 0$, $0<k\le d_{n}$, 
 and distinct letters $x,y\in \Sigma$ such that $xy=ab$ if $n$ is odd or $ba$ if $n$ is even.
\end{enumerate}
\end{lemma}

Note that $(u_{n}xy)^{d_{n}+1}u_{n-1}$ is a central prefix of $w$, but 
this does not contradict the previous lemma since, by~\eqref{eq:un+1},
$(u_nxy)^{d_n+1}u_{n-1}=u_{n+1}yxu_n$.

Recall that a \emph{semicentral word} (see~\cite{BuDelFi13}) is a word in which the longest repeated prefix, the longest repeated suffix, the longest left special factor and the longest right special factor all coincide.
The following proposition summarizes some properties of semicentral words proved in~\cite{BuDelFi13}.
\begin{proposition}
\label{prop:qxyq}
A word $v$ is semicentral if and only if $v=uxyu$ for a central word $u$ and distinct letters $x,y\in \Sigma$. Moreover, $u$ has exactly one internal occurrence in $v=uxyu$, and this occurrence is preceded by $x$ and followed by $y$. In particular, semicentral words are open (whereas central words are closed).
\end{proposition}

\begin{proposition}\label{prop:semi}
The semicentral prefixes of $w$ are precisely the words of the form $u_nxyu_n$, $n\ge 1$, where $x,y$ and $u_{n}$ are as in \eqref{eq:su}.
\end{proposition}

\begin{proof}
 Since $u_{n}$ is a central word, the word $u_nxyu_n$ is a semicentral word by definition, and it is a prefix of $u_nxyu_{n+1}=u_{n+1}yxu_{n}$, which in turn is a prefix of $w$ by Lemma \ref{lem:pref}.
 
 Conversely, assume that $w$ has a prefix of the form $u\xi\eta u$ for a central word $u$ and distinct letters $\xi,\eta\in \Sigma$. From Lemma \ref{lem:pref} and \eqref{eq:su}, we have that
 \[u\xi\eta u=(u_{n}xy)^{k}u_{n-1}\cdot \xi\eta \cdot (u_{n}xy)^{k}u_{n-1},\]
 for some $n\ge 1$, $k\le d_{n}$, and distinct letters $x,y\in \Sigma$ such that $xy=ab$ if $n$ is odd or $ba$ if $n$ is even. In particular, this implies that $\xi\eta=yx$.
 
If $k=d_{n}$, then $u=u_{n+1}yx u_{n+1}$, and we are done. So, suppose by contradiction that $k<d_{n}$. Now, on the one hand we have that $(u_{n}xy)^{k+1}u_{n-1}yx$ is a prefix of $w$ by Lemma \ref{lem:pref}, and so $(u_{n}xy)^{k+1}u_{n-1}$ is followed by $yx$ as a prefix of $w$; on the other hand we have
\begin{eqnarray*}
u\xi\eta u &=& (u_{n}xy)^{k}u_{n-1}\cdot yx \cdot (u_{n}xy)^{k}u_{n-1}  \\
&=& (u_{n}xy)^{k} \cdot u_{n-1}yxu_{n} xy \cdot(u_{n}xy)^{k-1}u_{n-1} \\
&=& (u_{n}xy)^{k} \cdot u_{n}xyu_{n-1} xy \cdot(u_{n}xy)^{k-1}u_{n-1} \\
&=& (u_{n}xy)^{k+1} \cdot u_{n-1}xy \cdot(u_{n}xy)^{k-1}u_{n-1},
\end{eqnarray*}
so that $(u_{n}xy)^{k+1}u_{n-1}$ is followed by $xy$ as a prefix of $w$, a contradiction.
\end{proof}

The next theorem shows the behavior of the runs in $\oc(w)$ by determining the structure of the last elements of the runs.

\begin{theorem}
\label{thm:bdaries}
Let $vx$, $x\in \Sigma$, be a prefix of $w$. Then:
\begin{enumerate}
\item $v$ is open and $vx$ is closed if and only if there exists $n\ge 1$ such that $v=u_{n}xyu_{n}$;
\item $v$ is closed and $vx$ is open if and only if there exists $n\ge 0$ such that $v=u_{n}xyu_{n+1}=u_{n+1}yxu_{n}$.
\end{enumerate}
\end{theorem}

\begin{proof}
1. If $v=u_{n}xyu_{n}$, then $v$ is semicentral and therefore open. The word $vx$ is closed since its longest repeated prefix $u_{n}x$ occurs only as a prefix and as a suffix in it.
 
Conversely, let $vx$ be a closed prefix of $w$ such that $v$ is open, and let $ux$ be the longest repeated suffix of $vx$. Since $vx$ is closed, $ux$ does not have internal occurrences in $vx$. Since $u$ is the longest repeated prefix of $v$ (suppose the longest repeated prefix of $v$ is a $z$ longer than $u$, then $vx$, which is a prefix of $z$, would be repeated in $v$ and hence in $vx$, contradiction) and $v$ is open, $u$ must have an internal occurrence in $v$ followed by a letter $y\neq x$. Symmetrically, if $\xi$ is the letter preceding the occurrence of $u$ as a suffix of $v$, since $u$ is the longest repeated suffix of $v$ one has that $u$ has an internal occurrence in $v$ preceded by a letter $\eta \neq \xi$. Thus $u$ is left and right special in $w$. Moreover, $u$ is the longest special factor in $v$. Indeed, if $u'$ is a left special factor of $v$, then $u$ must be a prefix of $u'$. But $ux$ cannot appear in $v$ since $vx$ is closed, and if $uy$ was a left special factor of $v$, it would be a prefix of $v$. Symmetrically, $u$ is  the longest right special factor in $v$. 
Thus $v$ is semicentral, and the claim follows from Proposition \ref{prop:semi}.
 
2. If $v=u_{n}xyu_{n+1}=u_{n+1}yxu_{n}$, then $v$ is a central word and therefore it is closed. Its longest repeated prefix is $u_{n+1}$. The longest repeated prefix of $vx$ is either $a^{d_0-1}$ (if $n=0$) or $u_{n}x$ (if $n>0$); in both cases, it has an internal occurrence as a prefix of the suffix $u_{n+1}x$. Therefore, $vx$ is open.
 
Conversely, suppose that $vx$ is any open prefix of $w$ such that $v$ is closed. If $vx=a^{d_0}b$, then $v=u_0xyu_1=u_1yxu_0$ and we are done. Otherwise, by 1), there exists $n\geq 1$ such that
$|u_n\xi y u_n|<|v|<|u_{n+1}y\xi u_{n+1}|$, where
$\{\xi,y\}=\{a,b\}$. We know that $u_n\xi y u_{n+1}$ is closed
and $u_n\xi y u_{n+1}\xi$ is open; it follows
$v=u_n\xi y u_{n+1}=u_nxyu_{n+1}$, as otherwise there should be in $w$ a semicentral prefix strictly between $u_nxyu_n$ and
$u_{n+1}yxu_{n+1}$.
\end{proof}

Note that, for every $n\ge 1$, one has:
\begin{eqnarray*}
  u_{n+1}yxu_{n+1}&=&u_{n+1}yxu_{n}(u_{n}^{-1}u_{n+1})\\
 &=&u_{n}xyu_{n+1}(u_{n}^{-1}u_{n+1})\\
 &=&u_{n}xyu_{n}(u_{n}^{-1}u_{n+1})^{2}.
\end{eqnarray*}
Therefore, starting from an (open) semi-central prefix $u_{n}xyu_{n}$, one has a run of closed prefixes, up to the prefix $u_{n}xyu_{n+1}=u_{n+1}yxu_{n}=u_{n}xyu_{n}(u_{n}^{-1}u_{n+1})$, followed by a run of the same length of open prefixes, up to the prefix $u_{n+1}yxu_{n+1}=u_{n+1}yxu_{n}(u_{n}^{-1}u_{n+1})=u_{n}xyu_{n}(u_{n}^{-1}u_{n+1})^{2}$. See Table \ref{tab:example} for an illustration.

\begin{table}[ht]
\setlength{\tabcolsep}{10pt}
%\begin{normalsize}
\begin{center}
\begin{tabular}{ l c l }
    prefix of $w$  &   open/closed & example  \\    \hline 
     $u_{n}xyu_{n}$           & open   & $aaba$        \\
     $u_{n}xyu_{n}x$          & closed   &    $aabaa$      \\   
     $u_{n}xyu_{n}xy$          & closed   & $aabaab$        \\
     \ldots      &   \ldots          & \ldots         \\
     $u_{n}xyu_{n+1}=u_{n+1}yxu_n$          & closed      & $aabaabaa$       \\
     $u_{n+1}yxu_{n}y$          & open   & $aabaabaaa$          \\
     $u_{n+1}yxu_{n}yx$           & open      & $aabaabaaab$       \\
	 \ldots      &   \ldots          & \ldots      \\
	 $u_{n+1}yxu_{n+1}$      & open   & $aabaabaaabaa$ \\
	 $u_{n+1}yxu_{n+1}y$        & closed   & $aabaabaaabaab$ \\
    \hline \\   
\end{tabular}
\end{center}
%\end{normalsize} 
\caption{The structure of the prefixes of the standard Sturmian word $w=aabaabaaabaabaa\cdots$ with respect to the $u_{n}$ prefixes. Here $d_{0}=d_{1}=2$ and $d_{2}=1$.\label{tab:example}}
\end{table}

In Table \ref{tab:oc}, we show the first few elements of the sequence $\oc(w)$ for the standard Sturmian word $w=aabaabaaabaabaa\cdots$ of slope $\alpha=(9+\sqrt{5})/38=\left[0;3,2,\bar{1}\right]$, i.e., with $d_{0}=d_{1}=2$ and $d_{i}=1$ for every $i>1$. One can notice that the runs of closed prefixes are followed by runs of the same length of open prefixes. 

\begin{table}[ht]
\begin{small}
\begin{center}
\begin{tabular}{r*{15}{@{\hspace{1.4em}}c}}
 $n$    & 1 & 2 & 3 & 4 & 5 & 6 & 7 &
8 & 9 & 10 & 11 & 12 & 13 & 14 & 15 \\ \hline\\[-1ex]
 $w$    & $a$ & $a$ & $b$ & $a$ & $a$ & $b$ & $a$ &
$a$ & $a$ & $b$ & $a$ & $a$ & $b$ & $a$ & $a$ \\
\hline \\[-1ex]
$\oc(w)$ & 1 & 1 & 0 & 0 & 1 & 1 & 1 & 1 & 0 & 0 & 0 & 0 & 1 & 1 & 1\\[0.5ex]
\hline \\
\end{tabular}
\end{center}
\end{small} 
\caption{\label{tab:oc}The oc-sequence of the word  $w= aabaabaaabaabaa\cdots$}
\end{table}

The words $u_{n}^{-1}u_{n+1}$ are reversals of standard words, for every $n\ge 1$. Indeed, let $r_n=\wt{s_n}$ for every $n\geq -1$, so that $r_{-1}=b$, $r_0=a$, and
$r_{n+1}=r_{n-1}r_n^{d_n}$ for $n\geq 0$. Since by \eqref{eq:su} $s_{n}=u_{n}xy$ and $s_{n+1}=u_{n+1}yx$, one has $r_n=yxu_{n}$ and $r_{n+1}=xyu_{n+1}$, and therefore, by \eqref{eq:u},
\begin{equation}\label{eq:ur}
 u_{n}r_{n+1}=u_{n+1}r_{n}.
\end{equation}
Multiplying \eqref{eq:ur}  on the left by $u_{n}^{-1}$ and on the right by $r_{n}^{-1}$, one obtains
\begin{equation}\label{eq:prop}
 r_{n+1}r_n^{-1}=u_n^{-1}u_{n+1}.
\end{equation}
Since $r_{n+1}=r_{n-1}r_{n}^{d_{n}}$, one has that $r_{n+1}r_{n}^{-1}=r_{n-1}r_{n}^{d_{n}-1}$, and therefore $r_{n+1}r_{n}^{-1}$ is the reversal of a standard word. By \eqref{eq:prop}, $u_n^{-1}u_{n+1}$ is the reversal of a standard word.

Now, note that for $n=0$, one has $u_{0}xyu_{1}=u_{1}yxu_{0}=a^{d_{0}}$ and $(u_{0}^{-1}u_{1})=ba^{d_{0}-1}$. Thus, we have the following:
\begin{theorem}
\label{theor:decomp}
Let $w$ be the standard Sturmian word of slope $\alpha$, with $0<\alpha<1/2$, and let $[0;d_{0}+1,d_{1},\ldots]$, with $d_{0}>0$, be the continued fraction expansion of $\alpha$. The word $ba^{-1}w$, obtained from $w$ by replacing the first letter $a$ with the letter $b$, can be written as an infinite product of squares of reversed standard words in the following way:
\[ba^{-1}w=\prod_{n\ge 0}(u_{n}^{-1}u_{n+1})^{2},\]
where $(u_{n})_{n\ge -1}$ is the sequence defined in \eqref{eq:su}.

In other words, one can write
\[w=a^{d_{0}}ba^{d_{0}-1}\prod_{n\ge 1}(u_{n}^{-1}u_{n+1})^{2}.\]
\end{theorem}
 
\begin{example}
Take the Fibonacci word. Then, $u_{1}=\epsilon$, $u_{2}=a$, $u_{3}=aba$, $u_{4}=abaaba$, $u_{5}=abaababaaba$, etc. So, $u_{1}^{-1}u_{2}=a$, $u_{2}^{-1}u_{3}=ba$, $u_{3}^{-1}u_{4}=aba$, $u_{4}^{-1}u_{5}=baaba$, etc. Indeed, $u_{n}^{-1}u_{n+1}$ is the reversal of the Fibonacci finite word $f_{n-1}$. By Theorem \ref{theor:decomp}, we have:
\begin{eqnarray*}
 F &=& ab\prod_{n\ge 1}(u_{n}^{-1}u_{n+1})^{2}\\
 &=& ab\prod_{n\ge 0}(\wt{f_{n}})^{2}\\
 &=& ab\cdot (a\cdot a)(ba \cdot ba)(aba \cdot aba)(baaba\cdot baaba)\cdots
\end{eqnarray*}
%$$F = ab\prod_{n\ge 1}(u_{n}^{-1}u_{n+1})^{2} = ab(a\cdot a)(ba \cdot ba)(aba \cdot aba)(baaba\cdot baaba)\cdots$$  
i.e., $F$ can be obtained by concatenating $ab$ and the squares of the reversals of the finite Fibonacci words $f_n$ starting from $n=0$.

Note that $F$ can also be obtained by concatenating the reversals of the finite Fibonacci words $f_n$ starting from $n=0$:
\begin{eqnarray*}
 F &=& \prod_{n\ge 0} \wt{f_{n}}\\
 &=& a \cdot ba \cdot aba \cdot baaba \cdot ababaaba \cdots
\end{eqnarray*}
%$$F = a \cdot ba \cdot aba \cdot baaba \cdot ababaaba \cdots$$ 
and also by concatenating $ab$ and the finite Fibonacci words $f_n$ starting from $n=0$:
\begin{eqnarray*}
 F &=& ab\prod_{n\ge 0} f_{n}\\
 &=& ab \cdot a \cdot ab \cdot aba \cdot abaab \cdot abaababa \cdots
\end{eqnarray*}
For a survey on various factorizations of the Fibonacci infinite word that make use of finite Fibonacci words the reader can see \cite{Fi15}.
\end{example}

One can also characterize the oc-sequence of a standard Sturmian word $w$ in terms of the  directive sequence of $w$.

Recall that the \emph{continuants} of an integer sequence $(a_n)_{n\geq 0}$ are defined as
$K\left[\ \, \right]=1$, $K\left[a_0\right]=a_0$, and, for every $n\geq 1$,
\[K\left[a_0,\ldots,a_n\right]=a_nK\left[a_0,\ldots,a_{n-1}\right]+K\left[a_0,\ldots,a_{n-2}\right].\] 
Continuants are related to continued fractions, as the $n$-th convergent of
$[a_0; a_1,a_2,\ldots]$ is equal to $K\left[a_0,\ldots, a_n\right]/K\left[a_1,\ldots, a_n\right]$.

Let $w$ be a standard Sturmian word and $(s_n)_{n\ge -1}$ its standard sequence. Since $|s_{-1}|=|s_0|=1$ and, for every $n\ge 1$,
$|s_{n+1}|=d_n|s_n|+|s_{n-1}|,$ then one has, by definition, that for every $n\geq 0$
\[|s_n|=K\left[1,d_0,\ldots, d_{n-1}\right].\]

For more details on the relationships between continuants and Sturmian words the reader can see \cite{dL13}.

By Theorems~\ref{thm:bdaries} and~\ref{theor:decomp}, all prefixes up to
$a^{d_0}$ are closed; then all prefixes from $a^{d_0}b$ till $a^{d_0}ba^{d_0-1}$ are open,
then closed up to $a^{d_0}ba^{d_0-1}\cdot u_1^{-1}u_2$, open again up to 
$a^{d_0}ba_{d_0-1}\cdot (u_1^{-1}u_2)^2$, and so on.
Thus, the lengths of the successive runs of closed and open prefixes are: 
$d_0$, $d_0$, $|u_2|-|u_1|$, $|u_2|-|u_1|$, $|u_3|-|u_2|$, $|u_3|-|u_2|$, etc. 
Since $d_0=K\left[1,d_0-1\right]$ and, for every $n\geq 1$,
\[
 |u_{n+1}|-|u_n|= |s_{n+1}|-|s_n|=(d_n-1)|s_n|+|s_{n-1}|
 =K\left[1,d_0,\ldots,d_{n-1},d_n-1\right],
\]
we have the following:

\begin{corollary}\label{cor:formula}
Let $w$ be a standard Sturmian word of slope $\alpha=[0;d_0+1,d_1,\ldots]$ and let $k_n=K\left[1,d_0,\ldots,d_{n-1},d_n-1\right]$ for every $n\geq 0$. Then
\[oc(w)=\prod_{n\geq 0}1^{k_n}0^{k_n}.\]
\end{corollary}

We now give a characterization of the prefixes of a standard Sturmian words in terms of the oc-sequence.

\begin{theorem}
\label{thm:ocst}
Let $\oc(w)=1^{k_0}0^{k'_0}1^{k_1}0^{k'_1}\cdots 1^{k_{n}}0^{k'_{n}}1$. Then $w$ is a prefix of a standard Sturmian word if and only if $k_j=k'_j$ for every $0\leq j\leq n$.
\end{theorem}
We need the following lemma.
\begin{lemma}
\label{lem:qocc}
Let $q$ be a central word and $\{x,y\}=\Sigma$. The word  $(qxy)^\omega$ has infinitely many prefixes ending in $xq$, and each of them is a central word of the form $(qxy)^np=p(yxq)^n$ for some $n>0$ and a central word $p$.
\end{lemma}

\begin{proof}
Let us consider the semicentral word $v=qxyq$. By Proposition~\ref{prop:qxyq}, $xq$ has exactly one (internal) occurrence in $v$. Therefore, the prefix $u$ of $v$ ending in $xq$ is a complete return to $q$ and hence it is closed, whereas the next prefix $uy$ of $v$ is open since $qy$ is not a prefix of $uy$. By Theorem~\ref{thm:bdaries}, it follows that $u$ is central, so that by Proposition~\ref{prop:PER} there exists a central word $p$ (shorter than $q$) such that $v=qxyp=pyxq$.

Thus, since every occurrence of $xq$ within $(qxy)^\omega$ is contained in a factor $qxyq$, it follows that any prefix of $(qxy)^\omega$ ending in $xq$ can be written as
\[(qxy)^np=(qxy)^{n-1}pyxq=\cdots=p(yxq)^n\,.\]
for some integer $n>0$.
\end{proof}

\begin{proof}[Proof of Theorem~\ref{thm:ocst}]
The ``only if'' part follows from Corollary~\ref{cor:formula}. Let us prove the ``if'' part
by induction on $n$.

For $n=0$, the statement is easily verified.
Let $\oc(w)=1^{k_0}0^{k_0}1^{k_1}0^{k_1}\cdots 1^{k_{n-1}}0^{k_{n-1}}1^{k_{n}}0^{k_{n}}1$ with $n>0$. By induction, we can suppose that the word $w'$ such that $\oc(w')=1^{k_0}0^{k_0}\cdots 1^{k_{n-1}}0^{k_{n-1}}1$ is a prefix of a standard Sturmian word. By Theorem \ref{thm:bdaries},  we can write $w'=qxyqx$, for a central word $q$ and distinct letters $x,y$. 

By Remark~\ref{rem:nth1}, any word $\hat w$ such that $\oc(\hat w)=\oc(w)$ is a complete return to its prefix $u$ of length $\sum_{i=0}^{n}k_i$. Since $|\hat w|=|w|=2\sum_{i=0}^{n}k_i+1$, it follows that $\hat w=u\xi u$ for some letter $\xi$. As $|qxyqx|=|w'|=2\sum_{i=0}^{n-1}k_i+1$, we have $|q|=\sum_{i=0}^{n-1}k_i-1$, so that $q$ is a prefix of $u$.
Now, $\oc(u\xi q)=1^{k_0}0^{k_0}\dotsm 1^{k_{n-1}}0^{k_{n-1}}1^{k_n}$, so that by Lemma~\ref{lem:ce} the word $u\xi q$ has the same period as $w'$ and hence is uniquely determined. This shows that
$\hat w=u\xi u=w$.

Since $qx$ is a prefix of $u$, $u\xi qx$ is a prefix of $u\xi u$. As $\oc(u\xi qx)=1^{k_0}0^{k_0}\dotsm 1^{k_{n-1}}0^{k_{n-1}}1^{k_n}0$, by Lemma~\ref{lem:ce} the period of $u\xi qx$ is different from the one of $u\xi q$, i.e., $|qxy|$. This implies that $\xi=x$, since otherwise $yqx$ would be a suffix of $u\xi qx$, so that $u\xi qx$ would still have period $|qxy|$.

By Lemma~\ref{lem:qocc}, $uxq$ is a central word that can be written as $uxq=(qxy)^jp=p(yxq)^j$ for some $j>0$ and a central word $p$. Hence we obtain $u=p(yxq)^{j-1}y$, so that \[w=uxu=p(yxq)^{j-1}yxp(yxq)^{j-1}y\,.\]

Thus, $p(yxq)^{j-1}=(qxy)^{j-1}p$ is a central word $r$, and $w=ryxry$ is a prefix of the infinite word $(ryx)^\omega$. Therefore, by Lemma~\ref{lem:qocc}, $w$ is a prefix of a central word and hence a prefix of a standard Sturmian word.   The proof is therefore complete.
\end{proof}

%%%%%%%%%%%%%%%%%%%%%%%%%%%%%%%%%%%%%%%%%%%%%%%%%%%%%%%%%%%%%%%%%%%%%%%%%%%%%%%%%%%%%%%%%%%%%%%%%%%%%%%
\section{Algorithms for the oc-sequence}
%%%%%%%%%%%%%%%%%%%%%%%%%%%%%%%%%%%%%%%%%%%%%%%%%%%%%%%%%%%%%%%%%%%%%%%%%%%%%%%%%%%%%%%%%%%%%%%%%%%%%%%

In this section, we show a simple linear-time algorithm that, given a word $w$ over a finite alphabet $\Sigma$, computes its sequence $\oc(w)$, and a linear-time algorithm that, given the sequence $\oc(w)$ of a Sturmian word $w$, reconstructs $w$.

Recall that the \emph{border array} $B(w)$ of the word $w$ is the sequence whose $i$-th entry is the length of the longest border of the prefix of length $i$ of $w$. For example, if $w=abcaacab$, then $B(w)=00011012$. We also define the array $B'(w)$ by $B'(w)[i]=\max_{j\leq i}B(w)[j]$.

\begin{proposition}\label{prop:algo}
Let $w$ be a nonempty word. Then $\oc(w)[1]=1$ and for every $i>0$, $\oc(w)[i]=B'(w)[i]-B'(w)[i-1]$.
\end{proposition}

\begin{proof}
By definition of a closed word, $\oc(w)[i]=1$ if and only if the longest border of the prefix of $w$ of length $i$ is longer than the border of any shorter prefix.
\end{proof}
\noindent As an example, for the word $w=abcaacab$ over the alphabet $\{a,b,c\}$, we have $B'(w)=00011112$, and indeed $\oc(w)=10010001$.

Since the border array of a word $w$ can be computed in linear time with respect to the length of $w$ \cite{MoPr70}, Proposition~\ref{prop:algo} gives a linear-time algorithm to compute $\oc(w)$.

\medskip
Suppose now that $w$ is a finite Sturmian word. By Theorem~\ref{theor:main}, it is possible to reconstruct $w$ from $\oc(w)$, up to isomorphisms of the alphabet $\Sigma=\{a,b\}$. In the following, we describe a linear-time algorithm (see Algorithm~\vref{alg:oc2s}) to do this.

Without loss of generality, assume $w[1]=a$. In order to obtain the whole of $w$, it is then
sufficient to calculate the border array $B$, as $B[i]<i$ and
%\begin{equation}
%\label{eq:copiadaB}
$w[i]=w[B[i]]$
%\end{equation}
hold for $i=2,\ldots,|w|$, provided that we extend $w$ to the left by setting $w[0]=b$ (see lines 6 and 21
in the algorithm). Now let $2\leq i\leq |w|$.
\begin{itemize}
\item If $\oc[i]=1$, by Proposition~\ref{prop:algo} it follows that $B[i]=B[i-1]+1$, i.e., the
$i$-th letter is the one that keeps the minimal period fixed (lines 11--12).
\item If $\oc[i]=0$ and $\oc[i-1]=1$, let $x=w[B[i-1]+1]$ and $\Sigma=\{x,y\}$.
We must have $w[i]=y$ since
otherwise $w[1\ldots i]$ would be a complete return to $vx$, with $v=w[1\ldots B[i-1]]$.
We can then recover $B[i]$ by the standard procedure~\cite{MoPr70}; this amounts to searching
for the longest border $u$ of $v$ that is followed by $y$ and then taking
$B[i]=|uy|$, or choosing $B[i]=|\varepsilon|=0$ if no such $u$ exists (lines 14--20).
\item For the next $0$s in a run, that is, when $\oc[i-1]=0$, we know that $v$ is a right
special factor of $w[1\ldots i-1]$, and no other factor of length $|v|$ is right special.
Therefore, up to the next occurrence of $v$, there is a unique letter $\xi$ that extends
$w[1\ldots i-1]$ to a Sturmian word. It is well known that extending a Sturmian
word by keeping the same period results in a Sturmian word (see for instance~\cite[Theorem~10]{DelDel06});
hence, if $v$ is not a suffix of
$w[1\ldots i-1]$, i.e., when $B[i-1]<|v|$, the letter $\xi$ is obtained by letting
$B[i]=B[i-1]+1$ (lines 11--12 again). When $B[i-1]=|v|$ and $\oc[i]=0$, the letter $\xi$ is
necessarily $y$, so that the longest border of $w[1\ldots i]$ is again $uy$ or 
$\varepsilon$ (line 20).
\end{itemize}

\lstset{language=Pascal,columns=flexible,mathescape=true,morekeywords={return},%
texcl,escapechar=`,commentstyle=\itshape\sffamily%}
,numbers=left,numberstyle=\scriptsize,xleftmargin=24pt}

\begin{lstlisting}[float={bth},caption={Function ReconstructSturmianWord.},label=alg:oc2s,%
belowcaptionskip=\medskipamount]
function ReconstructSturmianWord($\oc$)
{Input:		array $\oc=\oc[1\ldots n]$ of some Sturmian word
Output:		the corresponding Sturmian word $w=w[1\ldots n]$ starting with $a$}
begin
    $B$[0] := -1 ;
    $w$[0] := $b$ ;
    $B$[1] := 0 ;
    $w$[1] := $a$ ;
    ones := 0 ;
    for $i$ := 2 to $n$ do
        if $\oc[i]=1$ or $B[i-1]$ < ones then
            $B[i]$ := $B[i-1]+1$ ;
        else
            if $\oc[i-1]=1$ then
                ones := $B[i-1]$ ;
                $x$ := $w$[ones + 1] ;
                $j$ := $B$[ones] ;
                while $j\geq 0$ and $w[j+1] = x$ do
                    $j$ := $B[j]$ ;
            $B[i]$ := $j+1$ ;
        $w[i]$ := $w[B[i]]$ ;
    return $w:=w[1\ldots n]$ ;
end
\end{lstlisting}
%\captionof{figure}{Caption}

% 	{Conventional value added to border sequence $B$}
%			{Used for `\eqref{eq:copiadaB}` but not part of the output}
%		{Number of ones in $\oc[2\ldots i]$, only updated after each run}
%			{first $0$ in a run}
%			{keep the same period}
%		{finds $j=|u|$, with $uy$ border and $\{x,y\}=\Sigma$}
%			{before next $1$, $u$ is always followed by $y$}

%%%%%%%%%%%%%%%%%%%%%%%%%%%%%%%%%%%%%%%%%%%%%%%%%%%%%%%%%%%%%%%%%%%%%%%%%%%%%%%%%%%%%%%%%%%
\section{Conclusion and open problems}
%%%%%%%%%%%%%%%%%%%%%%%%%%%%%%%%%%%%%%%%%%%%%%%%%%%%%%%%%%%%%%%%%%%%%%%%%%%%%%%%%%%%%%%%%%%

In this paper we focused on the oc-sequence of a  word and exhibited results showing connections between this sequence and the combinatorics of the word. We mostly focused on Sturmian words, since these are characterized by their oc-sequence. Nevertheless, we believe that it may be interesting to also look at other classes of words. For example, in the case of the Tribonacci word $T=abacabaabacababacabaabac\cdots$, the sequence of the lengths of the runs of $1$ in $\oc(T)$ is exactly the Tribonacci sequence. We observed several regularities also in the oc-sequence of the Thue-Morse word, as well as in that of the regular paperfolding word. 

Another interesting problem is to understand, given a binary array $A$, whether there exists a word $w$ such that $\oc(w)=A.$ Some of the results in this paper provide necessary conditions, but the problem in general remains open. 

%%%%%%%%%%%%%%%%%%%%%%%%%%%%%%%%%%%%%%%%%%%%%%%%%%%%%%%%%%%%%%%%%%%%%%%%%%%%%%%%%%%%%%%%%%%
%\section*{Acknowledgments}
%%%%%%%%%%%%%%%%%%%%%%%%%%%%%%%%%%%%%%%%%%%%%%%%%%%%%%%%%%%%%%%%%%%%%%%%%%%%%%%%%%%%%%%%%%%

%\bibliographystyle{acm}
%
%\bibliography{biblio}

\end{document}